\documentclass[10pt,a4paper]{article}                                                    

\usepackage{setspace}                                                                    

\usepackage[top=2.0cm,bottom=1.9cm,left=1.5cm,right=1.5cm]{geometry}                     

\usepackage{layouts}                                                                     

\usepackage{multicol}                                                                    
\usepackage[small,bf]{caption}                                                           

\usepackage[T1]{fontenc} 
\usepackage{lmodern}     

\usepackage{scalerel}                                                                    

\usepackage{paralist}            
\setdefaultitem{-}{$\triangleright$}{}{}                                                 
\usepackage[ampersand]{easylist} 

\usepackage[usenames,dvipsnames]{xcolor}                                                 

\usepackage{amsfonts,amssymb}    
\usepackage[intlimits]{amsmath}  
\usepackage[]{mathalfa}
\usepackage{mathrsfs}            
\usepackage{bm}                  
\usepackage{upgreek}             
\usepackage{dsfont}              
\usepackage{mathtools}           
\usepackage{relsize}             
\usepackage{cancel}              
\usepackage{accents}             

\makeatletter
\newcommand{\raisemath}[1]{\mathpalette{\raisem@th{#1}}}
\newcommand{\raisem@th}[3]{\raisebox{#1}{$#2#3$}}
\makeatother

\usepackage{float}      
\usepackage{subfig}     
\usepackage{booktabs}   

\usepackage{amsthm}                                                                       
\newtheorem{defi}{Definition}[]                                                           
\newtheorem{theo}{Theorem}[]                                                              
\newtheorem{lemm}{Lemma}[]                                                                
\newtheorem{conj}{Conjecture}[]                                                           

\usepackage{algorithm}
\usepackage{algorithmic}

\definecolor{mycc}{rgb}{0, 0, 0.45}                                                       
\definecolor{mygreen}{rgb}{0, 0.4, 0}                                                     
\usepackage{ifpdf}                                                                        
\ifpdf                                                                                    
	\usepackage[pdftex=true,                                                                
		pdftitle={Stoch_iEnS},                                                                
		pdfauthor={Patrick N. Raanes},                                                        
		hyperindex=true,                                                                      
		colorlinks=true,                                                                      
		linkcolor=mygreen,                                                                    
		citecolor=mycc]{hyperref}                                                             
	\usepackage{epstopdf}                                                                   
	\usepackage[stretch=10]{microtype}                                                       
\else                                                                                     
	\usepackage[hypertex,                                                                    
		hyperindex=true,                                                                       
		colorlinks=false]{hyperref}                                                            
	\usepackage{graphicx}                                                                    
\fi                                                                                       

\graphicspath{{./}{Pics/}}

\hypersetup{                                                                             
}                                                                                        

\usepackage[colon,authoryear,square]{natbib}                                             

\usepackage{cleveref}
\crefname{equation}{equation}{equations}
\Crefname{equation}{Equation}{Equations}
\crefname{subeqns}{equation}{equations}
\Crefname{subeqns}{Equation}{Equations}
\crefformat{subeqns}{equation~(#2#1#3)}
\Crefformat{subeqns}{Equation~(#2#1#3)}
\labelcrefformat{subeqns}{(#2#1#3)}
\crefname{figure}{Figure}{Figures}
\Crefname{figure}{Figure}{Figures}
\crefformat{subfigure}{Figure~(#2#1#3)}
\Crefformat{subfigure}{Figure~(#2#1#3)}
\crefname{subfigure}{Figure}{Figures}
\Crefname{subfigure}{Figure}{Figures}
\crefname{table}{Table}{Tables}
\Crefname{table}{Table}{Tables}
\crefname{theo}{Theorem}{Theorems}
\Crefname{theo}{Theorem}{Theorems}
\crefname{theo}{Theorem}{Theorems}
\Crefname{theo}{Theorem}{Theorems}
\crefname{lemm}{Lemma}{Lemmas}
\Crefname{lemm}{Lemma}{Lemmas}
\crefname{conj}{Conjecture}{Conjectures}
\Crefname{conj}{Conjecture}{Conjectures}
\crefname{defi}{Definition}{Definitions}
\Crefname{defi}{Definition}{Definitions}
\crefname{prop}{Proposition}{Propositions}
\Crefname{prop}{Proposition}{Propositions}
\crefname{property}{Property}{Properties}
\Crefname{property}{Property}{Properties}
\crefname{coro}{Corollary}{Corollaries}
\Crefname{coro}{Corollary}{Corollaries}
\crefname{algocf}{Algorithm}{Algorithms}
\Crefname{algocf}{Algorithm}{Algorithms}

\usepackage{authblk}                                                                     

\newcommand{\NERSCa}{Thorm{\o}hlens gate 47, 5006 Bergen, Norway}
\newcommand{\NORCE}{NORCE}
\newcommand{\NORCEa}{Pb. 22 Nyg\aa rdstangen, 5838 Bergen, Norway}
\newcommand{\myemail}{patrick.n.raanes@gmail.com}

\title{Revising the stochastic iterative ensemble smoother}

\author[1,2]{Patrick Nima Raanes\thanks{                                                 
   \myemail                                                                              
   }}                                                                                    
	 \author[1,2]{Geir Evensen}                                                             
\author[1]{Andreas St\o rksen Stordal}                                                   
\affil[1]{\NORCE, \NORCEa         }                                                      
\affil[2]{NERSC,  \NERSCa         }                                                      

\newcommand{\tn}[1]{{\textnormal{{#1}}}}

\newcommand{\Reals}{\mathbb{R}}

\newcommand{\tq}[0]{\; : \;}

\newcommand{\Expect}[0]{\mathop{}\! \mathbb{E}}

\newcommand{\NormDist}{\mathop{}\! \mathcal{N}}

\DeclareMathAlphabet{\mathpzc}{OT1}{pzc}{m}{it}

\DeclareMathOperator{\col}{col}
\DeclareMathOperator{\rank}{rank}

\newcommand{\trsign}{{\mathsf{T}}}
\newcommand{\tr}{\ensuremath{^{\trsign}}}

\newcommand{\pinvsign}{{+}}

\newcommand{\pinv}{\ensuremath{^\pinvsign}}

\newcommand{\ones}[0]{\mathds{1}}

\newcommand*\diff{\mathop{}\!\mathrm{d}}
\newcommand*{\pdf}{\mathop{}\! p}

\newcommand{\dd}[2]{\frac{\diff #1}{\diff #2}}

\newcommand{\PDF}[0]{pdf}

\newcommand{\mat}[1]{{\mathbf{{#1}}}}
\newcommand{\bvec}[1]{{\bm{#1}}}

\newcommand{\x}[0]{\bvec{x}}
\newcommand{\w}[0]{\bvec{w}}
\newcommand{\e}[0]{\bvec{e}}
\newcommand{\y}[0]{\bvec{y}}
\newcommand{\z}[0]{\bvec{z}}
\newcommand{\bu}[0]{\bvec{u}}

\newcommand{\X}[0]{\mat{X}}
\newcommand{\A}[0]{\mat{A}}
\newcommand{\Y}[0]{\mat{Y}}
\newcommand{\E}[0]{\mat{E}}

\newcommand{\Ups}[0]{\mat{\Upsilon}}
\newcommand{\ups}[0]{\bvec{\upsilon}}

\newcommand{\Omeg}[0]{\mat{\Omega}}
\newcommand{\W}[0]{\mat{W}}
\newcommand{\T}[0]{\mat{T}}

\newcommand{\K}[0]{\mat{K}}
\newcommand{\I}[0]{\mat{I}}

\newcommand{\bS}[0]{\mat{S}}

\newcommand{\dtObs}{\Delta t_{\text{obs}} \,}
\newcommand{\dtDAW}{\Delta t_{\text{DAW}} \,}

\newcommand{\Z}[0]{\mat{Z}}
\newcommand{\D}[0]{\mat{D}}

\newcommand{\compactN}[0]{{N{-}1}}
\newcommand{\cN}[0]{(\compactN)}
\newcommand{\fracN}[0]{{{\tfrac{1}{N-1}}}} 

\newcommand{\Pro}[0]{\mat{\Pi}}
\newcommand{\PiOne}[0]{{\Pro_\ones}}
\newcommand{\PiAN}[0]{\Pro_\ones^\perp}

\newcommand{\RMSE}[0]{{\text{RMSE}}}

\newcommand{\tnf}[0]{\tn{prior}}
\newcommand{\tny}[0]{\tn{lklhd}}

\newcommand\matbar[1]{\accentset{\rule{.5em}{.09em}}{#1}}
\newcommand\vecbar[1]{\accentset{\rule{.4em}{.08em}}{#1}}

\newcommand{\ObsSym}      [0]{M}
\newcommand{\ObsMod}      [0]{\mathcal    {\ObsSym}}
\newcommand{\ObsMat}      [0]{\mat        {\ObsSym}}
\newcommand{\barObsMat}   [0]{\mat{\matbar{\ObsSym}}}
\newcommand{\barObsMatk}  [0]{\mat{\matbar{\ObsSym}}_i}
\newcommand{\barObsMatktr}[0]{\mat{\matbar{\ObsSym}}_i^{\raisemath{-0.5ex}{\trsign}}}

\newcommand{\Res}[0]{\mat{D}}

\newcommand{\bx}[0]{\bvec{{\vecbar{x}}}}
\newcommand{\barK}[0]{\mat{\matbar{K}}}

\newcommand {\mux} [0] {\bvec{\mu}_\x}

\newcommand{\Dta}[0]{\mat{\Delta}}
\newcommand{\barDta}[0]{\matbar{\Dta}}

\newcommand{\sW}{\scaleobj{0.8}{\W}}

\newcommand{\itn}{\times}

\newcommand { \xn      } [1] {          \x_{#1}        } 
\newcommand { \yn      } [1] {          \y_{#1}        }
\newcommand { \res     } [0] {    \bvec{\delta}        } 
\newcommand { \resn    } [1] {        \res_{#1}        }
\newcommand { \xnzero  } [0] {          \x_{n,0}       }
\newcommand { \xnk     } [0] {          \x_{n,i}       }
\newcommand { \xnl     } [0] {          \x_{n,i+1}     }
\newcommand { \xns     } [0] {          \x_{\itn}      }
\newcommand { \witn    } [0] {          \w_{\itn}      }
\newcommand { \Knk     } [0] {          \K_{\itn}      }
\newcommand { \Hnk     } [0] {     \ObsMat_{\itn}      }
\newcommand { \Dynk    } [0] {        \Dta_{\itn}^\tny }
\newcommand { \Dfnk    } [0] {        \Dta_{\itn}^\tnf }
\newcommand { \barDl   } [0] {     \barDta      ^\tny  }
\newcommand { \barDp   } [0] {     \barDta      ^\tnf  }
\newcommand { \barDlnk } [0] {     \barDta_{\itn}^\tny }
\newcommand { \barDpnk } [0] {     \barDta_{\itn}^\tnf }

\newcommand { \Jn      } [0] {           J_{\x,n}      }
\newcommand { \DJnk    } [0] { \nabla \! J_{\itn}      }

\newcommand { \Jwn     } [0] {           J_{\w,n}                         }
\newcommand { \nJf     } [0] { \nabla \! J_{\sW}^\tnf                     }
\newcommand { \nJy     } [0] { \nabla \! J_{\sW}^\tny                     }

\newcommand { \Cx      } [0] {      \mat{C }_\x                           }
\newcommand { \Chx     } [0] {      \mat{C }_{\ObsMod(\x),\x}             }
\newcommand { \Cy      } [0] {      \mat{C }_{\y}                         }
\newcommand { \Cr      } [0] {      \mat{C }_\res                         }
\newcommand { \Cw      } [0] {      \mat{\matbar{C} }_{\w}                }
\newcommand { \Cwk     } [0] {      \mat{\matbar{C} }_{\w,i}              }
\newcommand { \Cs      } [0] {      \mat{C }_{\itn}                       }
\newcommand { \barCx   } [0] {      \mat{\matbar{C} }_\x                  }
\newcommand { \barCxk  } [0] {      \mat{\matbar{C} }_{\x,i}              }

\newcommand { \barChx  } [0] {      \mat{\matbar{C} }_{\ObsMod(\x),\x}    }
\begin{document}

\newcommand{\makeAbstractCmd}[1]{
	\newcommand*{\myabstract}{
		\begin{abstract}
			#1
		\end{abstract}
	}
}
\makeAbstractCmd{
	Ensemble randomized maximum likelihood (EnRML) is an iterative (stochastic) ensemble smoother,
	used for large and nonlinear inverse problems,
	such as history matching and data assimilation.
	Its current formulation is overly complicated and
	has issues with computational costs, noise, and covariance localization,
	even causing some practitioners to omit crucial prior information.
	This paper
	resolves these difficulties
	and streamlines the algorithm,
	without changing its output.
	These simplifications are achieved through the careful treatment of
	the linearizations and subspaces.
	For example, it is shown
	(a) how ensemble linearizations relate to average sensitivity,
	and (b) that the ensemble does not lose rank during updates.
	The paper also draws significantly on the theory of
	the (deterministic) iterative ensemble Kalman smoother (IEnKS).
	Comparative benchmarks are obtained with the Lorenz-96 model
	with these two smoothers and the ensemble smoother using multiple data assimilation (ES-MDA).
}

\newcommand*{\draftV}{
\begin{center}
	\vspace{-2em}
	Draft version 8
\end{center}
}


\twocolumn[                                                                              
  \begin{@twocolumnfalse}                                                                
    \maketitle                                                                           
    \vspace{-1.5em}                                                                      
    \myabstract~\\                                                                       
  \end{@twocolumnfalse}                                                                  
]                                                                                        



\section{Introduction}                                                                    
\label{sec:Intro}
Ensemble (Kalman) smoothers are approximate methods used for
data assimilation (state estimation in geoscience),
history matching (parameter estimation for petroleum reservoirs),
and other inverse problems constrained by partial differential equations.
Iterative forms of these smoothers,
derived from optimization perspectives,
have proven useful in improving the estimation accuracy
when the forward operator is nonlinear.
Ensemble randomized maximum likelihood (EnRML)
is one such method.

This paper rectifies several conceptual and computational complications with EnRML,
detailed in \cref{sec:intro:EnRML}.
As emphasized in \cref{sec:intro:IEnKS},
these improvements are largely inspired by the theory of
the iterative ensemble Kalman smoother (IEnKS).
\emph{Readers unfamiliar with EnRML may jump to
the beginning of the derivation},
starting in \cref{sec:RML},
which defines the inverse problem
and the idea of the randomized maximum likelihood method.
\Cref{sec:EnRML} derives the new formulation of EnRML,
which is summarized by \cref{algo:GN_EnRML} of \cref{sec:algo}.
\Cref{sec:experiments} shows benchmark experiments
obtained with various iterative ensemble smoothers.
\Cref{sec:EnKF_proofs} provides proofs of some of the mathematical results used in the text.


\subsection{Ensemble randomized maximum likelihood (EnRML): obstacles}
\label{sec:intro:EnRML}
The Gauss-Newton variant of EnRML
was given by \citet{gu2007iterative,chen2012ensemble},
with an important precursor from \citet{reynolds2006iterative}.
This version explicitly requires the
ensemble-estimated ``model sensitivity'' matrix, herein denoted $\barObsMatk$.
As detailed in \cref{sec:EnRML},
this is problematic because $\barObsMatk$ is noisy
and requires the computation of the pseudo-inverse of the ``anomalies'', $\X_i\pinv$,
for each iteration, $i$.

A Levenberg-Marquardt variant was proposed in the landmark paper of
\citet[][]{chen2013levenberg}.
%
Its main originality is a partial resolution to the above issue by modifying the Hessian
(beyond the standard trust-region step regularization):
the prior ensemble covariance matrix is replaced by
the posterior covariance (of iteration $i$): $\barCx \leftarrow \barCxk$.
Now the Kalman gain form of the \emph{likelihood increment}
is ``vastly simplified'', because
the linearization $\barObsMatk$ only appears in the product
$\barObsMatk \barCxk \barObsMatktr$,
which does not require $\X_i\pinv$.
%
For the \emph{prior increment},
on the other hand,
the modification breaks its Kalman gain form.
Meanwhile, the precision matrix form, i.e. their equation (10),
is already invalid because it requires the inverse of $\barCxk$.
Still, in their equation (15),
the prior increment is formulated with an inversion in ensemble space,
and also unburdened 
of the explicit computation of $\barObsMatk$.
Intermediate explanations are lacking,
but could be construed to involve approximate inversions.
Another issue is that
the pseudo-inverse of $\barCx$ is now required (via $\X$),
and covariance localization is further complicated.

An approximate version was therefore also proposed
in which the prior mismatch term is omitted from the update formula altogether.
%
This is not principled,
and severely aggravates the chance of overfitting
and poor prediction skill.
Therefore, 
unless the prior mismatch term is relatively insignificant,
overfitting must be prevented by limiting the number of steps
or by clever stopping criteria.
Nevertheless, this version has received significant attention
in history matching.

This paper revises EnRML;
without any of the above tricks,
we formulate the algorithm such that
there is no explicit computation of $\barObsMatk$,
and show how the product $\barObsMatk \X$ may be computed
without any pseudo-inversions of the matrix of anomalies.
Consequently, the algorithm is simplified,
computationally and conceptually,
and there is no longer any reason to omit the prior increment.
Moreover, the Levenberg-Marquardt variant
is a trivial modification of the Gauss-Newton variant.
The above is achieved by improvements to the derivation,
notably by
\begin{inparaenum}[(a)]
\item improving the understanding of the sensitivity (i.e. linearizations) involved,
\item explicitly and rigorously treating issues of rank deficiency and subspaces, and
\item avoiding premature insertion of singular value decompositions (SVD).
\end{inparaenum}

\subsection{Iterative ensemble Kalman smoother (IEnKS)}
\label{sec:intro:IEnKS}

The contributions of this paper (listed by the previous paragraph) are original,
but draw heavily on the theory of
the IEnKS of
\citet{sakov2012iterative,bocquet2012combining,bocquet2014iterative}.
Relevant precursors include \citep{zupanski2005maximum},
as well as the iterative, extended Kalman filter \citep[e.g.][]{jazwinski1970stochastic}.

It is informally known that EnRML can be seen as a stochastic flavour of the IEnKS
\citep{sakov2012iterative}.
Indeed, while the IEnKS update
takes the form of a deterministic, ``square-root'' transformation,
based on a single objective function,
EnRML uses stochastic ``perturbed observations'',
associated with an ensemble of randomized objective functions.

Another notable difference is that
the IEnKS was developed in the atmospheric literature,
while EnRML was developed in the literature on subsurface flow.
Thus, typically, 
the IEnKS is applied to (sequential) state estimation problems
such as filtering for chaotic dynamical systems,
while EnRML is applied to (batch) parameter estimation problems,
such as nonlinear inversion for physical constants and boundary conditions.
For these problems,
EnRML is sometimes referred to as the iterative ensemble smoother (IES).
As shown by \citet{gu2007iterative}, however,
EnRML is easily reformulated for the sequential problem.
Vice-versa, the IEnKS may be formulated for the batch problem.

The improvements to the EnRML algorithm herein
render it very similar to the IEnKS, also in computational cost.
It thus fully establishes that EnRML is
the stochastic ''counterpart'' to the IEnKS.
In spite of the similarities,
the theoretical insights and comparative experiments of this paper
should make it interesting also for readers already familiar with the IEnKS.

\section{RML}
\label{sec:RML}
Randomized maximum likelihood (RML)
\citep{kitanidis1995quasi,oliver1996conditional,oliver2008inverse}
is an approximate solution approach to a class of inverse problems.
The form of RML described here is a simplification,
common for large inverse problems,
without the use of a correction step (such as Metropolis-Hastings).
This restricts the class of problems for which it is unbiased,
but makes it more tractable \citep{oliver2017metropolized}.
Similar methods were proposed and studied by
\citet{bardsley2014randomize,liu2017uncertainty,morzfeld2018variational}.

\subsection{The inverse problem}
\label{sec:RML:IP}
Consider the problem of
estimating the unknown, high-dimensional state (or parameter) vector $\x \in \Reals^M$,
given the observation $\y \in \Reals^P$.
It is assumed that
the (generic and typically nonlinear) forward observation process
may be approximated by a computational model, $\ObsMod$, so that
\begin{align}
	\y     &= \ObsMod(\x) + \res \, ,
	\label{eqn:HMMy}
\end{align}
where the error, $\res$, is random and
gives rise to a likelihood, $\pdf(\y|\x)$.

In the Bayesian paradigm,
prior information is quantified as a probability density function (\PDF{})
called the prior, denoted $\pdf(\x)$,
and the truth, $\x$, is considered a draw thereof.
The inverse problem then consists of computing and representing
the posterior which, in principle, is given by pointwise multiplication:
\begin{align}
	\pdf(\x|\y) \propto \pdf(\y|\x) \pdf(\x)
	\label{eqn:Bayes}
	\, ,
\end{align}
quantifying the updated estimation of $\x$.
Due to the noted high dimensionality and nonlinearity, this can be challenging,
necessitating approximate solutions.

The prior is assumed Gaussian,
with mean $\mux$ and covariance $\Cx$,
i.e.
\begin{align}
	\pdf(\x)
	&= \NormDist(\x \,|\, \mux,\Cx)
	\notag
	\\
	&= |2 \pi \Cx|^{-\frac{1}{2}}
	\, e^{- \frac{1}{2} \|\x - \mux\|^2_{\Cx}}
	\label{eqn:prior}
	\, .
\end{align}
\emph{For now}, the prior covariance matrix, $\Cx \in \Reals^{M \times M}$,
is assumed invertible
such that the corresponding norm, $\|\x\|_{\Cx}^2 = \x\tr\Cx^{-1}\x$, is defined.
Note that vectors are taken to have column orientation,
and that $\x\tr$ denotes the transpose.

The observation error, $\res$, is assumed drawn from:
\begin{align}
	\pdf(\res) &= \NormDist(\res \,|\, \bvec{0}, \Cr)
	\label{eqn:res}
	\, ,
\end{align}
whose covariance, $\Cr \in \Reals^{P \times P}$, will always be assumed invertible.
Then,
assuming $\res$ and $\x$ are independent
and recalling \cref{eqn:HMMy},
\begin{align}
	\pdf(\y|\x) &= \NormDist(\y \,|\, \ObsMod(\x), \Cr) \, .
	\label{eqn:lklhd}
\end{align}

\subsection{Randomize, then optimize}
\label{sec:RML:algo}

The Monte-Carlo approach offers a
convenient representation of distributions as samples.
Here, the prior is represented by the ``prior ensemble'', $\{\xn{n}\}_{n=1}^N$,
whose members (sample points) are assumed independently drawn from it.
RML
is an efficient method
to approximately ``condition''
(i.e. implement \labelcref{eqn:Bayes} on)
the prior ensemble,
using optimization.
Firstly, an ensemble of perturbed observations, $\{\yn{n}\}_{n=1}^N$,
is generated as $\yn{n} = \y + \resn{n}$,
where $\resn{n}$ is independently drawn according to \cref{eqn:res}.

Then, the $n$-th ``randomized log-posterior'',
$\Jn$,
is defined by Bayes' rule \labelcref{eqn:Bayes},
except with the prior mean and the observation
replaced by the $n$-th members
of the prior and observation ensembles:
\begin{align}
	\label{eqn:Jn}
	\Jn(\x)
	&=
	{\textstyle \frac{1}{2}}
	\|\x - \xn{n}\|^2_{\Cx}
	+
	{\textstyle \frac{1}{2}}
	\|\ObsMod(\x) - \yn{n}\|^2_{\Cr}
	\, .
\end{align}
The two terms are referred to as
the model mismatch (log-prior) and data mismatch (log-likelihood),
respectively.

Finally, these log-posteriors are minimized.
Using the Gauss-Newton iterative scheme (for example)
requires
\labelcref{eqn:jac} its gradient and
\labelcref{eqn:hess} its Hessian approximated by first-order model expansions,
both evaluated at the current iterate,
labelled $\xnk$ for each member $n$ and iteration $i$.
To simplify the notation,
define $\xns = \xnk$.
Objects evaluated at $\xns$ are
similarly denoted;
for instance, $\Hnk = \ObsMod'(\xns) \in \Reals^{P \times M}$ denotes
the Jacobian of $\ObsMod$ evaluated at $\xns$,
and
\begin{subequations}
	\begin{align}
		\label[equation]{eqn:jac}
		\DJnk &= \Cx^{-1}[ \xns - \xn{n}] + \Hnk\tr \Cr^{-1} [\ObsMod( \xns ) - \yn{n}]
		\, ,
		\\
		\label[equation]{eqn:hess}
		\Cs^{-1} &= \Cx^{-1} + \Hnk\tr \Cr^{-1} \Hnk
		\, .
	\end{align}
	\label[subeqns]{eqn:Jppp}%
\end{subequations}
Application of the Gauss-Newton scheme yields:
\begin{equation}
	\begin{split}
		\xnl
		&=
		\xns - \Cs \, \DJnk
		\\
		&=
		\xns + \Dfnk + \Dynk
		\label{eqn:E_DD}
		\, ,
	\end{split}
\end{equation}
where
the prior (or model) and likelihood (or data) increments are
respectively given by:
\begin{subequations}
	\begin{align}
		\label[equation]{eqn:b_pri}
		\Dfnk &= \Cs \Cx^{-1}[\xn{n} - \xns]
		\, ,
		\\
		\label[equation]{eqn:b_lkle}
		\Dynk &= \Cs \Hnk\tr \Cr^{-1} [\yn{n} - \ObsMod( \xns )]
		\, ,
	\end{align}
	\label[subeqns]{eqn:D_RML}%
\end{subequations}
which can be called the ``precision matrix'' form.

Alternatively, by corollaries of the well known Woodbury matrix identity,
the increments can be written in the ``Kalman gain'' form:
\begin{subequations}
	\begin{align}
		\label[equation]{eqn:c_pri}
		\Dfnk &= (\I_M - \Knk \Hnk)[\xn{n} - \xns]
		\, ,
		\\
		\label[equation]{eqn:c_lkl}
		\Dynk &= \Knk [\yn{n} - \ObsMod( \xns )]
		\, ,
	\end{align}
	\label[subeqns]{eqn:D_RML_2}%
\end{subequations}
where $\I_M \in \Reals^{M \times M}$ is the identity matrix,
and 
$\Knk \in \Reals^{M \times P}$ is the gain matrix:
\begin{align}
	\Knk &= \Cx \Hnk\tr  \Cy^{-1}
	\, ,
\end{align}
with
\begin{align}
	\Cy &= \Hnk \Cx \Hnk\tr + \Cr
	\label{eqn:Cy}
	\, .
\end{align}
As the subscript suggests, $\Cy$ may be identified (in the linear case)
as the prior covariance of the observation, $\y$, of \cref{eqn:HMMy};
it is also the covariance of the innovation, $\y - \ObsMod (\mux)$.
Note that if 
$P \ll M$,
then the inversion of $\Cy \in \Reals^{P \times P}$ for the Kalman gain form \labelcref{eqn:D_RML_2}
is significantly cheaper than the inversion of $\Cs \in \Reals^{M \times M}$
for the precision matrix form \labelcref{eqn:D_RML}.


\section{EnRML}
\label{sec:EnRML}

Ensemble-RML (EnRML) is an approximation of RML in which
the ensemble is used in its own update,
by estimating $\Cx$ and $\Hnk$.
This section derives EnRML,
and gradually introduces the new improvements.

Computationally, compared to RML,
EnRML offers the simultaneous benefits of
working with low-rank representations of covariances,
and not requiring a tangent-linear (or adjoint) model.
Both advantages will be further exploited in the new formulation of EnRML.

Concerning their sampling properties,
a few points can be made.
Firstly (due to the ensemble covariance),
EnRML is biased for finite $N$,
even for a linear-Gaussian problem,
for which RML will sample the posterior correctly.
This bias arises for the same reasons as in
the ensemble Kalman filter \citep[EnKF,][]{van1999comment,sacher2008sampling}.
Secondly (due to the ensemble linearization),
EnRML effectively smoothes the likelihood.
It is therefore less prone to getting trapped
in local maxima of the posterior \citep{chen2012ensemble}.
\citet{sakov2018iterative} explain this by drawing an analogy to
the secant method, as compared to the Newton method.
Hence, it may reasonably be expected that 
EnRML yields constructive results
if the probability mass of the exact posterior
is concentrated around its global maximum.
Although this regularity condition is rather vague,
it would require that the model be ``not too nonlinear'' in this neighbourhood.
Conversely, EnRML is wholly inept at reflecting multimodality
introduced through the likelihood,
and so RML may be better suited when local modes feature prominently,
as is quite common in problems of subsurface flow \citep{oliver2011recent}.
However, while RML has the ability to sample multiple modes,
it is difficult to predict to what extent their relative proportions will be accurate
(without the costly use of a correction step such as Metropolis-Hastings).
Further comparison of the sampling properties of RML and EnRML
was done by \citet{evensen2018analysis}.

\subsection{Ensemble preliminaries}
\label{sec:en_prelim}

\newcommand{\xnw  }[1]{\makebox[1.5em][c]{$\xn  {#1}$}}
\newcommand{\resnw}[1]{\makebox[1.5em][c]{$\resn{#1}$}}
For convenience, define the concatenations:
\begin{align}
	\label{eqn:concat_E}
	\E &= \begin{bmatrix}
		\xnw{1}, & \ldots & \xnw{n}, & \ldots & \xnw{N}
	\end{bmatrix} \in \Reals^{M \times N} \, ,
	\\
	\Res &=
	\begin{bmatrix}
		\resnw{1}, & \ldots & \resnw{n}, & \ldots & \resnw{N}
	\end{bmatrix}
	\in \Reals^{P \times N}
	\label{eqn:concat_D}
	\, ,
\end{align}
which are known as the ``ensemble matrix'' and the ``perturbation matrix'',
respectively.
%

Projections sometimes appear
through the use of linear regression.
We therefore recall \citep[][]{trefethen1997numerical} that
a (square) matrix $\Pro$ is an orthogonal projector if
\begin{align}
	\Pro \Pro = \Pro = \Pro\tr
	\label{eqn:Pro_def}
	\, .
\end{align}
For any matrix $\A$, let $\Pro_{\A}$ denote
the projector whose image is the column space of $\A$,
implying that
\begin{align}
	\Pro_{\A} \A &= \A
	\, .
\end{align}
Equivalently, $\Pro_{\A}^\perp \A = \mat{0}$,
where $\Pro^\perp_{\A} = \I - \Pro_{\A}$ is called the complementary projector.
The (Moore-Penrose) pseudo-inverse, $\A\pinv$,
may be used to express the projector:
\begin{align}
	\Pro_{\A} = \A \A\pinv = (\A\tr)\pinv (\A\tr)
	\label{eqn:ProjPinv}
	\, .
\end{align}
Here, the second equality follows from the first
by \cref{eqn:Pro_def} and $(\A\pinv)\tr = (\A\tr)\pinv$.
The formulae simplify further in terms of the SVD of $\A$.

Now, denote $\ones \in \Reals^{N}$ the (column) vector of ones.
%
The matrix of anomalies, $\X \in \Reals^{M \times N}$, is defined and computed by
subtracting the ensemble mean,
$\bx = \E \ones/N$, from each column of $\E$.
It should be appreciated that this amounts to the projection:
\begin{align}
	\X &= \E - \bx \ones\tr = \E \PiAN \, ,
	\label{eqn:A1}
\end{align}
where
$\PiAN = \I_N - \PiOne$,
with $\PiOne = \ones\ones\tr / N$.

\begin{defi}[The ensemble subspace]
	The flat (i.e. affine subspace) given by:
	$\{\x \in \Reals^M \tq [\x - \bx] \in \col(\X)\}$.
\end{defi}
Similarly to \cref{sec:RML},
iteration index ($i>0$) subscripting on $\E$, $\X$,
and other objects,
is used to indicate that they are conditional (i.e. posterior).
The iterations are initialized with the prior ensemble: $\xnzero = \xn{n}$.

\subsection{The constituent estimates}
\label{sec:underlying}
The ensemble estimates of $\Cx$ and $\Hnk$ are the building blocks of the EnRML algorithm.
The canonical estimators are used,
namely the sample covariance \labelcref{eqn:e},
and the least-squares linear regression coefficients \labelcref{eqn:f}.
They are denoted with the overhead bar:
\begin{subequations}
	\begin{align}
		\label[equation]{eqn:e}
		\barCx   &= \fracN \X \X\tr \, , \\
		\label[equation]{eqn:f}
		\barObsMatk
		&= \ObsMod(\E_i) \X_i\pinv
		\, .
	\end{align}
	\label[subeqns]{eqn:B_H_bar}%
\end{subequations}
The anomalies at iteration $i$ are again given by $\X_i = \E_i \PiAN$,
usually computed by subtraction of $\bx_i$.
The matrix $\ObsMod(\E_i)$ is defined by
the column-wise application of $\ObsMod$ to the ensemble members.
Conventionally, $\ObsMod(\E_i)$ would also be centred in \cref{eqn:f},
i.e. multiplied on the right by $\PiAN$.
However, this operation (and notational burden) can be neglected,
because $\PiAN \X_i\pinv = \X_i\pinv$,
which follows from $\Pro (\A \Pro)\pinv = (\A \Pro)\pinv$
\citep[valid for any matrix $\A$ and projector $\Pro$, as shown by][]{maciejewski1985obstacle}.

Note that the linearization (previously $\Hnk$, now $\barObsMatk$) no longer depends on the ensemble index, $n$.
Indeed, it has been called ``average sensitivity''
since the work of \citet{zafari2005assessing,reynolds2006iterative,gu2007iterative}.
However, this intuition has not been rigorously justified.\footnote{%
	The formula \labelcref{eqn:f} for $\barObsMatk$
	is sometimes arrived at via a truncated Taylor expansion 
	of $\ObsMod$ around $\bx_i$.
	This is already an approximation,
	and still requires further, indeterminate approximations
	to obtain any other interpretation than
	$\ObsMod' (\bx_i)$: the Jacobian evaluated at the ensemble mean.
}
This is accomplished by the following theorem.
\begin{theo}[Regression coefficients versus derivatives]
	\label{theo:sens}
	Suppose the ensemble is drawn from a Gaussian.
	Then
	\begin{align}
		\lim_{N \rightarrow \infty} \barObsMat
		&= \Expect [\ObsMod'(\x)]
		\, ,
	\end{align}
	with ``almost sure'' convergence,
	and expectation $(\Expect)$ in $\x$,
	which has the same distribution as the ensemble members.
	Regularity conditions and proof in \cref{sec:EnKF_proofs}.
\end{theo}
A corollary of \cref{theo:sens} is that $\barObsMat \approx \frac{1}{N} \sum_{n=1}^N \ObsMod'(\x_n)$,
justifying the ``average sensitivity/derivative/gradient'' description.
The theorem applies for the ensemble of \emph{any} Gaussian,
and hence also holds for $\barObsMatk$.
On the other hand, the generality of \cref{theo:sens} is restricted
by the Gaussianity assumption.
Thus, for generality and precision,
$\barObsMatk$ should simply be labelled
``the least-squares (linear) fit'' of $\ObsMod$, based on $\E_i$.
%

Note that the computation \labelcref{eqn:f} of $\barObsMatk$
	seemingly requires calculating a new pseudo-inverse, $\X_i\pinv$, at each iteration, $i$;
	this is addressed in \cref{sec:Y}.

The prior covariance estimate (previously $\Cx$, now $\barCx$) is \emph{not} assumed invertible,
in contrast to \cref{sec:RML}.
It is then not possible to employ the precision matrix forms \labelcref{eqn:D_RML}
because $\barCx^{-1}$ is not defined.
Using the $\barCx\pinv$ in its stead is flawed and damaging because it
is zero in the directions orthogonal to the ensemble subspace,
so that its use would imply that the prior is assumed
infinitely uncertain (i.e. flat) 
as opposed to infinitely certain (like a delta function) in those directions.
Instead,
one should employ ensemble subspace formulae,
or equivalently (as shown in the following, using corollaries of the Woodbury identity),
the Kalman gain form.

\subsection{Estimating the Kalman gain}
The ensemble estimates \labelcref{eqn:B_H_bar}
are now substituted into 
the Kalman gain form of the update, \cref{eqn:D_RML_2} to \labelcref{eqn:Cy}.
The ensemble estimate of the gain matrix, denoted $\barK_i$,
thus becomes:
\begin{align}
	\barK_i
	&= \barCx \barObsMatktr \big(\barObsMatk \barCx \barObsMatktr + \Cr\big)^{-1} \notag \\
	\label{eqn:barK_Y}
	&= \X \Y_i\tr \big(\Y_i \Y_i\tr + \cN \Cr\big)^{-1}
	\, ,
\end{align}
where $\Y_i \in \Reals^{P \times N}$ has been defined as the \emph{prior} (i.e. unconditioned) anomalies,
under the action of the $i$-th iterate linearization:
\begin{align}
	\label{eqn:Yk0_def}
	\Y_i &= \barObsMatk \X
	\, .
\end{align}
A Woodbury corollary
can be used to express $\barK_i$ as:
\begin{align}
	\label{eqn:l}
	\barK_i
	&= \X \Cwk \Y_i\tr \Cr^{-1}
	\, ,
\end{align}
with
\begin{align}
	\label{eqn:tP}
	\Cwk &= \big(\Y_i\tr \Cr^{-1} \Y_i + \cN \I_N\big)^{-1}
	\, .
\end{align}
The reason for labelling this matrix with the subscript $\w$ is revealed later.
For now, note that, in the common case of $N \ll P$,
the inversion in \cref{eqn:tP} is significantly cheaper than the inversion in \cref{eqn:barK_Y}.
Another computational benefit is that $\Cwk$ is non-dimensional,
improving the conditioning of the optimization problem \citep{lorenc1997development}.

In conclusion,
the likelihood increment \labelcref{eqn:c_lkl}
is now estimated as:
\begin{align}
	\barDlnk &= \barK_i [\yn{n} - \ObsMod( \xns )]
	\label{eqn:lklhd_inc_bar}
	\, .
\end{align}
This is efficient because
$\barObsMatk$ does not explicitly appear
in $\barK_i$ (neither in formula \labelcref{eqn:barK_Y} nor \labelcref{eqn:l}),
even though it is implicitly present through $\Y_i$ \labelcref{eqn:Yk0_def},
where it multiplies $\X$.
This absence
\begin{inparaenum}[(a)]
\item is reassuring, as the product $\Y_i$ constitutes
a less noisy estimate than just $\barObsMatk$ alone
\citep[][figures 2 and 27, resp.]{chen2012ensemble,emerick2013investigation};
\item constitutes a computational advantage,
as will be shown in \cref{sec:Y};
\item
enables leaving
the type of linearization made for $\ObsMod$ unspecified,
as is usually the case in EnKF literature.
\end{inparaenum}

\subsection{Estimating the prior increment}
\label{sec:prior_inc}
In contrast to the likelihood increment \labelcref{eqn:c_lkl},
the Kalman gain form of the prior increment \labelcref{eqn:c_pri}
explicitly contains the sensitivity matrix, $\Hnk$.
This issue was resolved by \citet{bocquet2012combining}
in their refinement of \citet{sakov2012iterative}
by employing the change of variables:
\begin{align}
	\x(\w) = \bx + \X \w
	\label{eqn:x_CVar}
	\, ,
\end{align}
\noindent
where $\w \in \Reals^N$ is called
the ensemble ``controls'' \citep{bannister2016review}, 
also known as the ensemble ``weights'' \citep{ott2004local},
or ``coefficients'' \citep{bocquet2013joint}.

Denote $\witn$ an ensemble coefficient vector such that
$\x(\witn) = \xns$,
and note that $\x(\e_n) = \xn{n}$,
where $\e_n$ is the $n$-th column of the identity matrix.
Thus, $[\xn{n} - \xns] = \X [ \e_n - \witn ]$,
and the prior increment \labelcref{eqn:c_pri}
with the ensemble estimates becomes:
\begin{align}
	\barDpnk
	&= (\X - \barK_i \Y_i) [\e_n - \witn]
	\label{eqn:prior_inc_bar}
	\, ,
\end{align}
where there is no explicit $\barObsMatk$,
which only appears implicitly through $\Y_i = \barObsMatk \X$,
as defined in \cref{eqn:Yk0_def}
Alternatively,
applying the subspace formula \labelcref{eqn:l}
and using $\I_N = \Cwk (\Cwk)^{-1}$
yields:
\begin{align}
	\barDpnk
	&= \X \Cwk \cN [\e_n - \witn]
	\, .
	\label{eqn:prior_inc_Pw}
\end{align}

\subsection{Justifying the change of variables}
\label{sec:cvar_comment}

\begin{lemm}[Closure]
	Suppose $\E_i$ is generated by EnRML.
	Then, each member (column) of $\E_i$ is in the (prior) ensemble subspace.
	Moreover, $\col(\X_i) \subseteq \col(\X)$.
	\label{lemm:Ak_space}
\end{lemm}
\Cref{lemm:Ak_space} may be proven
by noting that $\X$ is the leftmost factor in $\barK_i$,
and using induction on \cref{eqn:c_pri,eqn:c_lkl}.
Alternatively, it can be deduced \citep{raanes2019adaptive}
as a consequence of the
implicit assumption on the prior that
$\x \sim \NormDist( \bx, \barCx )$.
A stronger result, namely $\col(\X_i) = \col(\X)$,
is conjectured in \cref{sec:EnKF_proofs},
but \cref{lemm:Ak_space} is sufficient for the present purposes:
it implies that there exists $\witn \in \Reals^N$ such that $\x(\witn) = \xns$
for any ensemble member and any iteration.
Thus, the lemma justifies the change of variables \labelcref{eqn:x_CVar}.


Moreover, using the ensemble coefficient vector ($\w$) is theoretically advantageous
as it inherently embodies the restriction to the ensemble subspace.
A practical advantage is that $\w$ is relatively low-dimensional compared to $\x$,
which lowers storage and accessing expenses.

\subsection{Simplifying the regression}
\label{sec:Y}
Recall the definition of \cref{eqn:Yk0_def}: $\Y_i = \barObsMatk \X$.
Avoiding the explicit computation of $\barObsMatk$ used in this product between
the iteration-$i$ estimate $\barObsMatk$ and the initial (prior) $\X$
was the motivation behind the modification
$\barCx \leftarrow \barCxk$ by \citet{chen2013levenberg}.
Here, instead, by simplifying the expression of the regression,
it is shown how to compute $\Y_i$ without first computing $\barObsMatk$.


\subsubsection{The transform matrix}
\label{sec:T}
Inserting the regression $\barObsMatk$ \labelcref{eqn:f} into the definition \labelcref{eqn:Yk0_def},
\begin{align}
	\Y_i
	\label{eqn:m_2}
	&= \ObsMod(\E_i) \, \T_i\pinv
	\, ,
\end{align}
where $\T_i\pinv = \X_i\pinv \X$ has been defined,
apparently requiring the pseudo-inversion of $\X_i$ for each $i$.
But, as shown in \cref{sec:Tinv},
\begin{align}
	\T_i &= \X\pinv \X_i
	\, ,
	\label{eqn:T_def}
\end{align}
which only requires the one-time pseudo-inversion of the prior anomalies, $\X$.
Then, since the pseudo-inversion of $\T_i \in \Reals^{N \times N}$
for $\Y_i$ \labelcref{eqn:m_2}
is a relatively small calculation,
this saves computational time.

The symbol $\T$ has been chosen in reference to deterministic, square-root EnKFs.
Indeed, multiplying \cref{eqn:T_def} on the left by $\X$
and recalling \cref{eqn:ProjPinv,lemm:Ak_space}
produces $\X_i = \X \T_i$.
Therefore, the ``transform matrix'', $\T_i$,
describes the conditioning of the anomalies (and covariance).

Conversely,
\cref{eqn:m_2} can be seen as the ``de-conditioning'' of the posterior observation anomalies.
This interpretation of $\Y_i$ should be contrasted to its definition \labelcref{eqn:Yk0_def},
which presents it as the prior state anomalies ``propagated'' by the linearization of iteration $i$.
The two approaches are known to be ``mainly equivalent''
in the deterministic case \citep{sakov2012iterative}.
To our knowledge, however,
it has not been exploited for EnRML before now,
possibly because the proofs (\cref{sec:Tinv})
are a little more complicated in this stochastic case.

\subsubsection{From the ensemble coefficients}
\label{sec:T_from_W}
The ensemble matrix of iteration $i$ can be written:
\begin{align}
	\E_i = \bx \ones\tr + \X \W_i
	\, ,
	\label{eqn:E_CVar}
\end{align}
where the columns of $\W_i \in \Reals^{N \times N}$ are
the ensemble coefficient vectors \labelcref{eqn:x_CVar}.
Multiplying \cref{eqn:E_CVar} on the right
by $\PiAN$ to get the anomalies produces:
\begin{align}
	\X_i = \X (\W_i \PiAN) \, .
	\label{eqn:XkW}
\end{align}
This seems to indicate that $\W_i \PiAN$ is the transform matrix, $\T_i$,
discussed in the previous subsection.
However, they are not fully equal:
inserting $\X_i$ from \labelcref{eqn:XkW} into \labelcref{eqn:T_def} yields:
\begin{align}
	\T_i
	&= 
	\Pro_{\X\tr} (\W_i \PiAN)
	\, ,
	\label{eqn:T_PWP}
\end{align}
i.e. they are distinguished by $\Pro_{\X\tr} = \X\pinv \X$:
the projection onto the row space of $\X$.

\Cref{sec:red_T} shows that,
in most conditions,
this pesky projection matrix vanishes when
$\T_i$ is used in \cref{eqn:m_2}:
\begin{align}
	\Y_i
	&=
	\ObsMod(\E_i) \, (\W_i \PiAN)\pinv
	\quad \text{if} 
	\begin{dcases}
		\compactN \leq M, \text{ or} \\
		\text{$\ObsMod{}$ is linear.}
	\end{dcases}
	\label{eqn:Yred}
\end{align}
In other words, the projection $\Pro_{\X\tr}$ can be omitted
unless $\ObsMod{}$ is nonlinear \emph{and} the ensemble is larger than the unknown state's dimensionality.

A well known result of \citet{reynolds2006iterative} is that
the first step of the EnRML algorithm (with $\W_0 = \I_N$)
is equivalent to the EnKF.
However, this is only strictly true if
there is no appearance of $\Pro_{\X\tr}$ in EnRML.
The following section explains why
EnRML should indeed always be defined without this projection.

\subsubsection{Linearization chaining}
\label{sec:chain_reg}
Consider applying the change of variables \labelcref{eqn:x_CVar} to $\w$
at the very beginning of the derivation of EnRML.
Since $\X \ones = 0$, there is a redundant degree of freedom in $\w$,
meaning that there is a choice to be made in deriving its density from 
the original one, given by $\Jn(\x)$ in \cref{eqn:Jn}.
The simplest choice \citep[][]{bocquet2015expanding} results in the log-posterior:
\begin{align*}
	\Jwn(\w)
	&=
	{\textstyle \frac{1}{2}}
	\|\w - \e_n\|^2_{\frac{1}{N-1}\I_N}
	+
	{\textstyle \frac{1}{2}}
	\| \ObsMod(\bx{+}\X \w) - \yn{n} \|^2_{\Cr}
	\, .
\end{align*}
Application of the Gauss-Newton scheme with the gradients and Hessian of $\Jwn$,
followed by a reversion to $\x$,
produces the same EnRML algorithm as above.

The derivation summarized in the previous paragraph
is arguably simpler than that of the last few pages.
Notably,
\begin{inparaenum}[(a)]
\item it does not require the Woodbury identity to derive the subspace formulae;
\item there is never an explicit $\barObsMatk$ to deal with;
\item the statistical linearization
of least-squares regression from $\W_i$ to $\ObsMod(\E_i)$
directly yields \cref{eqn:Yred},
except that there are no preconditions.
\end{inparaenum}

While the case of a large ensemble ($\compactN > M$) is not typical in geoscience,
the fact that this derivation does not produce a projection matrix (which requires a pseudo-inversion)
under any conditions begs the questions:
Why are they different? Which version is better?

The answers lie in understanding the linearization of
the map
$\w \mapsto \ObsMod(\bx + \X \w)$,
and noting that,
similarly to analytical (infinitesimal) derivatives,
the chain rule applies for least-squares regression.
In effect,
the product $\Y_i = \barObsMatk \X$,
which implicitly contains the projection matrix $\Pro_{\X\tr}$,
can be seen as
an application of the chain rule for the composite function $\ObsMod (\x (\w))$.
By contrast, \cref{eqn:Yred} -- but without the precondition --
is obtained by direct regression of the composite function.
Typically, the two versions yield identical results
(i.e. the chain rule).
However, since the intermediate space, $\col(\X)$,
is of lower dimensions than the initial domain ($M < \compactN$),
composite linearization results in a loss of information,
manifested by the projection matrix.
Therefore,
the definition $\Y_i = \ObsMod(\E_i) \, (\W_i \PiAN)\pinv$ is henceforth preferred to $\barObsMatk \X$.

Numerical experiments, as in \cref{sec:experiments} but not shown,
indicate no statistically significant advantage for either version.
This corroborates similar findings by \citet{sakov2012iterative}
for the deterministic flavour.
Nevertheless, there is a practical advantage:
avoiding the computation of $\Pro_{\X\tr}$.

\subsubsection{Inverting the transform}
\label{sec:Omeg}
In square-root ensemble filters, the transform matrix should have $\ones$ as an eigenvector
\citep{sakov2008implications,livings2008unbiased}.
By construction, this also holds true for $\W_i \PiAN$, 
with eigenvalue 0.
Now, consider adding $\mat{0} = \X \PiOne$ to \cref{eqn:XkW},
yielding another valid transformation:
\begin{align}
	\X_i &= \X (\underbrace{\W_i \PiAN + \PiOne}_{\Omeg_i})
	\label{eqn:fawef}
	\, .
\end{align}
The matrix $\Omeg_i$, in contrast to $\W_i \PiAN$ and $\T_i$,
has eigenvalue 1 for $\ones$ and is thus invertible.
This is used to prove \cref{eqn:Yred} in \cref{sec:red_T},
where $\Y_i$ is expressed in terms of $\Omeg_i^{-1}$.

Numerically,
the use of $\Omeg_i$ in the computation \labelcref{eqn:Yred} of $\Y_i$
was found to yield stable convergence of the new EnRML algorithm
in the trivial example of $\ObsMod(\x) = \alpha \x$.
By contrast, the use of $(\W \PiAN)\pinv$
exhibited geometrically growing (in $i$) errors when $\alpha>1$.
Other formulae for the inversion are derived in \cref{sec:pinv_v};
the one found to be the most stable is $(\W \PiAN)\pinv = \W^{-1} \PiAN$;
it is therefore preferred in \cref{algo:GN_EnRML}.

Irrespective of the inverse transform formula used,
it is important to retain all non-zero singular values.
This absence of a truncation threshold is a tuning simplification
compared with the old EnRML algorithm,
where $\X$ and/or $\X_i$ was scaled, decomposed, and truncated.
If, by extreme chance or poor numerical subroutines,
the matrix $\W_i$ is not invertible
(this never occurred in any of the experiments except by our explicit intervention;
cf. the conjecture in \cref{sec:EnKF_proofs}),
its pseudo-inversion should be used;
however, this must also be accounted for in the prior increment
by multiplying the formula on line \ref{ln:Jb} on the left
by the projection onto $\W_i$.

\subsection{Algorithm}
\label{sec:algo}

To summarize, \cref{algo:GN_EnRML} provides pseudo-code
for the new EnRML formulation.
The increments
$\barDl$ \labelcref{eqn:lklhd_inc_bar} and
$\barDp$ \labelcref{eqn:prior_inc_Pw}
can be recognized by multiplying line \ref{ln:Wk} on the left by $\X$.
For aesthetics, the sign of the gradients has been reversed.
Note that there is no need for an explicit iteration index.
Nor is there an ensemble index, $n$,
since all $N$ columns are stacked into the matrix $\W$.
However, in case $M$ is large,
$\Y$ may be computed column-by-column to avoid storing $\E$.
\begin{algorithm}[H]
	\caption{Gauss-Newton variant of EnRML \\
	(the stochastic flavour of the IEnKS analysis update)}
	\label[algocf]{algo:GN_EnRML} 
	\begin{algorithmic}[1]
		\REQUIRE prior ens. $\E$, obs. perturb's $\Res$
		\newlength{\mylength}
		\settowidth{\mylength}{$\W$}
		\STATE \makebox[\mylength][l]{$\bx$} $ = \E \ones /N $
		\STATE \makebox[\mylength][l]{$\X$}  $ = \E - \bx \ones\tr $
		\STATE \makebox[\mylength][l]{$\W$}  $ = \I_N$
		\REPEAT
			\settowidth{\mylength}{$\nJy$}
			\STATE Run model (on each col.) to get $\ObsMod(\E)$
      \STATE \makebox[\mylength][l]{$\Y$} $ = \ObsMod(\E) \, \W^{-1} \PiAN$
			\label{ln:Y}
			\STATE \makebox[\mylength][l]{$\nJy$} $ = \Y\tr \Cr^{-1} [\y \ones\tr  + \Res - \ObsMod(\E)] $
			\label{ln:Jl}
			\STATE \makebox[\mylength][l]{$\nJf$} $ = \cN [\I_N - \W]$
			\label{ln:Jb}
			\STATE \makebox[\mylength][l]{$\Cw$} $ = \big(\Y\tr \Cr^{-1} \Y + \cN \I_N\big)^{-1}$
			\label{ln:tPk}
			\STATE \makebox[\mylength][l]{$\W$} $ = \W + \Cw [\nJf + \nJy]$
			\label{ln:Wk}
			\STATE \makebox[\mylength][l]{$\E$} $ = \bx \ones\tr + \X \W$
			\UNTIL{tolerable convergence or max. iterations}
			\RETURN posterior ensemble $\E$
	\end{algorithmic}
\end{algorithm}

~\\ 
Line \ref{ln:Y} is typically computed by solving
$\Y' \W = \ObsMod(\E)$ for $\Y'$ and then subtracting its column mean.
Alternative formulae are discussed in \cref{sec:Omeg}.
Line \ref{ln:tPk} may be computed using a reduced (or even truncated) SVD of $\Cr^{-1/2} \Y$,
which is relatively fast for $N$ both larger and smaller than $P$.
Alternatively, the Kalman gain forms could be used.

The Levenberg-Marquardt variant is obtained by adding the trust-region parameter
$\lambda > 0$ to $\cN$ in the Hessian, line \ref{ln:tPk},
which impacts both the step length and direction.

Localization may be implemented by local analysis
\citep{hunt2007efficient,sakov2011relation};
also see \citet{bocquet2016localization,chen2017localization}.
Here, tapering is applied by replacing the local-domain $\Cr^{-1/2}$
(implicit on lines \ref{ln:Jl} and \ref{ln:tPk})
by $\rho \circ \Cr^{-1/2}$, with $\circ$ being the Schur product,
and $\rho$ a square matrix containing the (square-root) tapering coefficients,
$\rho_{m,l} \in [0,1]$.
If the number of local domains used is large,
so that the number of $\W$ matrices used becomes large,
then it may be more efficient to revert to the original state variables,
and explicitly compute the sensitivities $\barObsMatk$
using the local parts of $\ObsMod(\E_i)$ and $\X_i$.

Inflation and model error parameterizations are not included in the algorithm,
but may be applied outside of it.
We refer to \citet{sakov2018iterative,evensen2018accounting}
for model error treatment with iterative methods.

\section{Benchmark experiments}
\label{sec:experiments}

The new EnRML algorithm produces results that are \emph{identical} to the old formulation,
at least up to round-off and truncation errors, and for $N-1 \leq M$.
Therefore, since there are already a large number of studies of EnRML with reservoir cases
\citep[e.g.][]{chen2013history,emerick2013investigation},
adding to this does not seem necessary.

However, there do not appear to be any studies of EnRML with the  
Lorenz-96 system \citep{lorenz1996predictability}
in a data assimilation setting.
The advantages of this case are numerous:
\begin{inparaenum}[(a)]
\item the model is a surrogate of weather dynamics, and as such holds relevance in geoscience;
\item the problem is (exhaustively) sampled from the system's invariant measure,
	rather than being selected by the experimenter;
\item the sequential nature of data assimilation inherently tests prediction skill,
	which helps avoid the pitfalls of point measure assessment, such as overfitting;
\item its simplicity enhances reliability and reproducibility,
	and has made it a literature standard,
	thus facilitating comparative studies.
\end{inparaenum}

Comparison of the benchmark performance of EnRML
will be made to the IEnKS,
and to ensemble multiple data assimilation (ES-MDA)\footnote{%
	Note that this is MDA in the sense of
	\citet{emerick2013ensemble,stordal2015iterative,kirkpatrick1983optimization},
	where the annealing itself yields iterations,
	and not in the sense of quasi-static assimilation
	\citep{pires1996extending,bocquet2014iterative,fillon2018quasi},
	where it is used as an auxiliary technique.
}.
Both the stochastic and the deterministic (square-root) flavours of ES-MDA
are included,
which in the case of only one iteration (not shown),
result in exactly the same ensembles as EnRML and IEnKS, respectively.
Not included in the benchmark comparisons
is the version of EnRML where the prior increment is dropped (cf. \cref{sec:intro:EnRML}).
This is because the chaotic, sequential nature of this case
makes it practically impossible to achieve good results without propagating prior information.
Similarly, as they lack a dynamic prior, this precludes
``regularizing, iterative ensemble smoothers''
\citep{iglesias2015iterative},
\citep{luo2015iterative},\footnote{%
	Their Lorenz-96 experiment only concerns the initial conditions.}
\citep{mandel2016hybrid}\footnote{%
	Their Lorenz-96 experiment seems to have failed completely,
	with most of the benchmark scores (their Figure 5)
	indicating divergence,
	which makes it pointless to compare benchmarks.
	Also, when reproducing their experiment,
	we obtain much lower scores than they report for the EnKF.
	One possible explanation is that we include, and tune, inflation.
},
even if their background is well-tuned, and their stopping condition judicious.
Because they require the tangent-linear model, $\Hnk$,
RML and EDA/En4DVar \citep{tian2008ensemble,bonavita2012use,jardak2018ensemble1}
are not included.
For simplicity, localization will not be used,
nor covariance hybridization.
Other, related methods may be found in the reviews of
\citet{bannister2016review,carrassi2018da}.


\subsection{Setup}
\label{sec:setup}
The performances of the iterative ensemble smoother methods
are benchmarked with ``twin experiments'',
using the Lorenz-96 dynamical system,
which is configured with standard settings \citep[e.g.][]{ott2004local,bocquet2014iterative},
detailed below.
The dynamics are given by
the $M=40$ coupled ordinary differential equations:
\begin{align}
	\dd{x_m}{t}
	=
	\left( x_{m+1} - x_{m-2} \right) x_{m-1} - x_{m} + F \, ,
	\label{eqn:L96_1}
\end{align}
for $m = 1,\ldots,M$,
with periodic boundary conditions.
These are integrated using the fourth-order Runge-Kutta scheme,
with time steps of 0.05 time units,
and no model noise,
to yield the truth trajectory, $\x(t)$.
Observations of the entire state vector
are taken $\dtObs$
time units apart
with unit noise variance,
meaning $\y(t) = \x(t) + \res(t)$,
for each $t = k \cdot \dtObs$, with $k = 0,1,\ldots,20\,000$,
and $\Cr=\I_M$.

The iterative smoothers are employed in the sequential problem of filtering,
aiming to estimate $\x(t)$ as soon as $\y(t)$ comes in.
In so doing, they also tackle the smoothing problem for $\x(t{-}\dtDAW)$,
where the length of the data assimilation window,
$\dtDAW$, is fixed at a near-optimal value
\citep[inferred from Figures 3 and 4 of][]{bocquet2013joint}
that is also cost efficient (i.e. short).
This window is shifted by $1 \cdot \dtObs$ each time a new observation becomes available.
A post-analysis inflation factor is tuned
for optimal performance for each smoother and each ensemble size, $N$.
Also, random rotations are used to generate the ensembles for the square-root variants.
The number of iterations is fixed, either at $3$ or $10$.
No tuning of the step length is undertaken:
it is $1/3$ or $1/10$ for ES-MDA,
and $1$ for EnRML and the IEnKS.

The methods are assessed by their accuracy,
as measured by root-mean squared error:
\begin{align}
	\RMSE(t) = \sqrt{\frac{1}{M} \big\|\x(t) - \bx(t)\big\|^2_2}
	\, ,
	\label{eqn:RMSE_3}
\end{align}
which is recorded immediately following each analysis of the latest observation $\y(t)$.
The ``smoothing'' error [assessed with $\x(t{-}\dtDAW)$] is also recorded.
After the experiment, the instantaneous $\RMSE(t)$
are averaged for all $t>20$.
The results can be reproduced using Python-code scripts hosted online at
\url{https://github.com/nansencenter/DAPPER/tree/paper_StochIEnS}.
This code reproduces previously published results in the literature.
For example, our benchmarks obtained with the IEnKS
can be cross-referenced with the ones reported by \citet[][Figure 7a]{bocquet2014iterative}.

\subsection{Results}
\label{sec:Results}
\begin{figure*}
	\includegraphics[trim={0.0cm 0cm 0.0cm 0cm},clip]{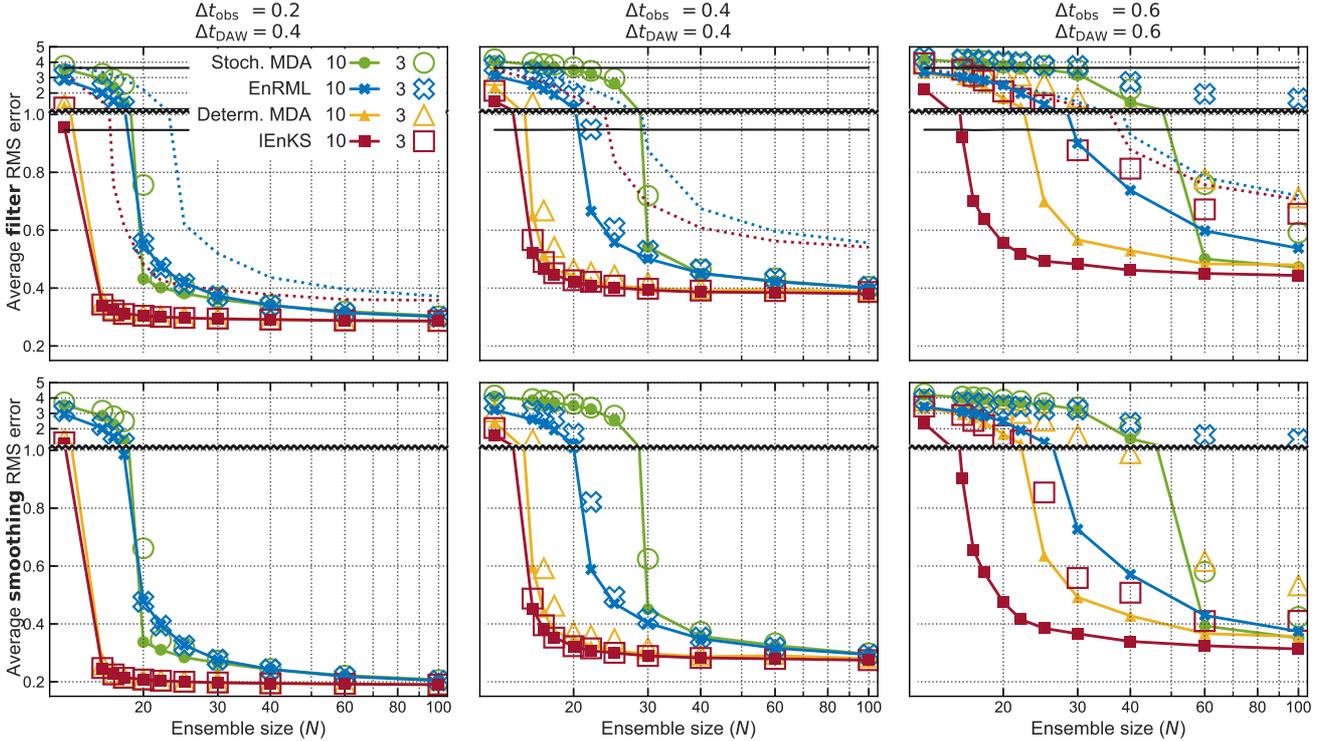}
	\caption
	{
		Benchmarks of filtering (upper panels) and smoothing (lower panels) accuracy,
		in three configurations of the Lorenz-96 system,
		plotted as functions of $N$.
		The $y$-axis changes resolution at $y=1$.
		Each iterative ensemble smoother (coloured, solid line) is plotted
		for 3 (hollow markers) and 10 (compact markers) iterations.
		It can be seen that the deterministic (i.e. square-root) methods systematically 
		achieve lower RMSE averages.
		For perspective, the black lines at $y = 3.6$ and $y = 0.94$
		show the average RMSE scores of the climatological mean,
		and of the optimal interpolation method,
		respectively.
		The dotted lines show the scores of the stochastic (blue) and deterministic (red) EnKF.
	}
	\label{fig:benchmarks}
\end{figure*}

A table of RMSE averages is compiled for a range of $N$,
and then plotted as curves for each method, in \cref{fig:benchmarks}.
The upper panels report the analysis RMSE scores,
while the lower panels report the smoothing RMSE scores.
The smoothing scores are systematically lower,
but the relative results are highly similar. 
Moving right among the panels
increases $\dtObs$,
and thus the nonlinearity;
naturally, all of the RMSE scores also increase.
As a final ``sanity check'', note that
the performances of all of the ensemble methods improve with increasing $N$,
which needs to be at least $15$ for tolerable performance,
corresponding to the rank of the unstable subspace of the dynamics
plus $1$ \citep{bocquet2016aus4d}.

For experiments with $\dtObs \leq 0.4$,
using $3$ iterations is largely sufficient,
since its markers are rarely significantly higher
than those of $10$ iterations.
On the other hand, for the highly nonlinear experiment where $\dtObs = 0.6$,
there is a significant advantage in using $10$ iterations.

The deterministic (square-root) IEnKS and ES-MDA
score noticeably lower RMSE averages than
the stochastic IEnKS (i.e. EnRML) and ES-MDA,
which require $N$ closer to $30$ for good performance.
This is qualitatively the same result as obtained for non-iterative methods
\citep[e.g.][]{sakov2008implications}.
Also tested (not shown) was the first-order-approximate
deterministic flavour of ES-MDA \citep{emerick2018deterministic};
it performed very similarly to the square-root flavour.

Among the stochastic smoothers,
the one based on Gauss-Newton (EnRML)
scores noticeably lower averages than the one based on annealing (ES-MDA)
-- when the nonlinearity is strong ($\dtObs \geq 0.4$), and for small $N$.
A similar trend holds for the deterministic smoothers:
the IEnKS performs better than ES-MDA for $\dtObs = 0.6$.
The likely explanation for this result is that EnRML/IEnKS can iterate indefinitely,
while ES-MDA may occasionally suffer from not ``reaching'' the optimum.

Furthermore, the performance of EnRML/IEnKS 
could possibly be improved by lowering the step lengths,
to avoid causing ``unphysical'' states,
and to avoid ``bouncing around'' near the optimum.
The tuning of the parameter that controls the step length,
(e.g. the trust-region parameter and the MDA-inflation parameter)
has been the subject of several studies
\citep{chen2012ensemble,bocquet2012combining,ma2017robust,le2016adaptive,rafiee2017theoretical}.
However, our superficial trials with this parameter (not shown) yielded little or no improvement.

\section{Summary}                                                                         
\label{sec:Summary}

This paper has presented a new and simpler (on paper and computationally) formulation of
the iterative, stochastic ensemble smoother known as ensemble randomized maximum likelihood (EnRML).
Notably, there is no explicit computation of the sensitivity matrix $\barObsMatk$,
while the product $\Y_i = \barObsMatk \X$ is computed without any pseudo-inversions of the matrix of state anomalies.
This fixes issues of noise, computational cost, and covariance localization,
and there is no longer any temptation to omit the prior increment from the update.
Moreover, the Levenberg-Marquardt variant
is now a trivial modification of the Gauss-Newton variant.

The new EnRML formulation was obtained
by improvements to the background theory and derivation.
Notably,
\cref{theo:sens} established the relation of
the ensemble-estimated, least-squares linear regression coefficients, $\barObsMatk$,
to ``average sensitivity''.
\Cref{sec:Y} then showed that the computation of its action on the prior anomalies,
$\Y_i = \barObsMatk \X$, simplifies into a de-conditioning transformation, $\Y_i = \ObsMod(\E_i) \, \T_i\pinv$.
Further computational gains resulted from
expressing $\T_i$ in terms of the coefficient vectors, $\W_i$,
except that it also involves the ``annoying'' $\Pro_{\X\tr}$.
Although it usually vanishes, the appearance of this projection
is likely the reason why most expositions of the EnKF
do not venture to declare that its implicit linearization of $\ObsMod$
is that of least-squares linear regression.
\Cref{sec:chain_reg} showed that the projection is
merely the result of using the chain rule for indirect regression to the ensemble space,
and argued that it is preferable to use the direct regression of the standard EnKF.

The other focus of the derivation was rank issues,
with $\barCx$ not assumed invertible.
Using the Woodbury matrix lemma,
and avoiding implicit pseudo-inversions and premature insertion of SVDs,
it was shown that the rank deficiency invalidates the Hessian form of the RML update,
which should be restricted to the ensemble subspace.
On the other hand, the subspace form and Kalman gain form of the update
remain equivalent and valid.
Furthermore, \cref{theo:W_full} of \cref{sec:EnKF_proofs}
proves that the ensemble does not lose rank during the updates of EnRML (or EnKF).

The paper has also drawn significantly on the theory of
the deterministic counterpart to EnRML:
the iterative ensemble Kalman smoother (IEnKS).
Comparative benchmarks using the Lorenz-96 model
with these two and the ensemble multiple data assimilation (ES-MDA) smoother
were shown in \cref{sec:experiments}.
In the case of small ensembles and large nonlinearity,
EnRML (resp. IEnKS) achieved better accuracy
than stochastic (resp. deterministic) ES-MDA.
Similarly to the trend for non-iterative filters,
the deterministic smoothers systematically obtained better accuracy
than the stochastic smoothers.


\appendix

\section{Proofs}
\label{sec:EnKF_proofs}

\subsection{Preliminary}

\begin{proof}[Proof of \cref{theo:sens}]
	Assume $0 < |\Cx| < \infty$,
	and that each element of $\Chx$ and $\Expect [\ObsMod'(\x)]$ is finite.
	Then $\barCx$ is a strongly consistent estimator of $\Cx$.
	Likewise,
	$\barChx \rightarrow \Chx$ almost surely, as $N \rightarrow \infty$.
	Thus, since $\barObsMat = \barChx \, \barCx^{-1}$ for sufficiently large $N$,
	Slutsky's theorem yields
	$\barObsMat \rightarrow \Chx \, \Cx^{-1}$,
	almost surely.
	%
	The equality to $\Expect [\ObsMod'(\x)]$
	follows directly from ``Stein's lemma'' \citep{liu1994siegel}.
\end{proof}

\begin{theo}[EnKF rank preservation]
	The posterior ensemble's covariance,
	obtained using the EnKF, has the same rank as the prior's,
	almost surely (a.s.).
	\label{theo:W_full}
\end{theo}
\begin{proof}
	The updated anomalies,
	both for the square-root and the stochastic EnKF,
	can be written
	$\X^a = \X \T^a$
	for some $\T^a \in \Reals^{N \times N}$.
   
	For a deterministic EnKF, 
  $\T^a = \sqrt{N-1}\Cw^{-1/2}$ for the symmetric positive definite square root of $\Cw$,
	or an orthogonal transformation thereof \citep[][]{sakov2008implications}.
	Hence $\rank(\X^a) = \rank(\X)$.

	For the stochastic EnKF, \cref{eqn:lklhd_inc_bar,eqn:l} may be used to show that
	$\T^a = \cN \Cw \Ups \PiAN$,
	with $\Ups = \I_N + \Y\tr \Cr^{-1} \D / \cN$.
	Hence, for rank preservation, it will suffice to show that $\Ups$ is a.s. full rank.

	We begin by writing $\Ups$ more compactly:
	\begin{align}
		\Ups = \I_N + \bS\tr \Z
		\quad \text{ with }
		\begin{dcases}
			\bS = \cN^{-1/2} \Cr^{-1/2} \Y \, , \\
			\Z = \cN^{-1/2} \Cr^{-1/2} \D \, .
		\end{dcases}
		\label{eqn:Ups}
	\end{align}
	\newcommand{\SSn}[1]{\mathcal{S}_{#1}}%
	From \cref{eqn:res,eqn:concat_D,eqn:Ups} it can be seen that
	column $n$ of $\Z$ follows the law $\z_n \sim \NormDist(\bvec{0}, \I_P/\cN)$.
	Hence, column $n$ of $\Ups$ follows $\ups_n \sim \NormDist(\e_n, \bS\tr \bS/\cN)$,
	and has sample space:
	\begin{align}
		\SSn{n} = \{\ups \in \Reals^N \tq \ups = \e_n + \bS\tr \z \}
		\label{eqn:SSn}
		\, .
	\end{align}
	%
	%
	Now consider, for $n=0,\ldots,N$,
	the hypothesis:
	\begin{align}
		\rank([\Ups_{:n},\ \I_{n:}]) = N
		\, , \tag{$\tn{H}_n$}
	\end{align}
	where $\Ups_{:n}$ denotes the first $n$ columns of $\Ups$,
	and $\I_{n:}$ denotes the last $N-n$ columns of $\I_N$.
	Clearly, $\tn{H}_{0}$ is true.
	Now, suppose $\tn{H}_{n-1}$ is true.
	Then the columns of $[\Ups_{:n-1},\ \I_{n-1:}]$ are all linearly independent.
	For column $n$, this means that
	$\e_n \notin \col([\Ups_{:n-1}, \ \I_{n:}])$.
	By contrast, from \cref{eqn:SSn}, $\e_n \in \SSn{n}$.
	The existence of a point in $\SSn{n} \setminus \col([\Ups_{:n-1}, \ \I_{n:}])$
		means that
	\begin{align}
		\label{eqn:Hn5}
		\dim\big( \SSn{n} \cap &\col([\Ups_{:n-1}, \ \I_{n:}]) \big) < \dim (\SSn{n})
		\, .
	\end{align}
	Since $\ups_n$ is absolutely continuous with sampling space $\SSn{n}$,
	\cref{eqn:Hn5} means that	the probability that
	$\ups_n \in \col([\Ups_{:n-1}, \ \I_{n:}])$ is zero.
	This implies $\tn{H}_{n}$ a.s.,
	establishing the induction.
	Identifying the final hypothesis ($\tn{H}_{N}$) with $\rank(\Ups) = N$
  concludes the proof.
\end{proof}
A corollary of \cref{theo:W_full} and \cref{lemm:Ak_space}
is that the ensemble subspace is also unchanged by the EnKF update.
Note that both the prior ensemble and the model (involved through $\Y$)
are arbitrary in \cref{theo:W_full}.
However, $\Cr$ is assumed invertible.
The result is therefore quite different from the topic discussed by 
\citet{kepert2004ensemble,evensen2004sampling},
where rank deficiency arises due to a reduced-rank $\Cr$.


\begin{conj}
	\label{conj:rank_EnRML}
	The rank of the ensemble is preserved by the EnRML update (a.s.)
	and $\W_i$ is invertible.
\end{conj}
We were not able to prove \cref{conj:rank_EnRML},
but it seems a logical extension of \cref{theo:W_full},
and is supported by numerical trials.
The following proofs utilize \cref{conj:rank_EnRML},
without which some projections will not vanish.
Yet, even if \cref{conj:rank_EnRML} should not hold
(due to bugs, truncation, or really bad luck),
\cref{algo:GN_EnRML} is still valid and optimal,
as discussed in \cref{sec:chain_reg,sec:Omeg}.

\subsection{The transform matrix}
\label{sec:Tinv}

\begin{theo}[]
	$(\X\pinv \X_i)\pinv = \X_i\pinv \X$.
	\label{lemm:T_pinv}
\end{theo}
\begin{proof}
	Let $\T = \X\pinv \X_i$ and $\bS = \X_i\pinv \X$.
	The following shows that $\bS$ satisfies
	the four properties of the Moore-Penrose characterization of the pseudo-inverse of $\T$:
	\begin{enumerate}
		\item \label{item:1}
			$\begin{aligned}[t]
				\T \bS \T
				&=
				( \X\pinv \X_i )
				( \X_i\pinv \X )
				( \X\pinv \X_i )
				\\
				&=
				\X\pinv
				\Pro_{\X_i}
				\Pro_{\X}
				\X_i
				&[\Pro_\A = \A\A\pinv]
				\\
				&=
				\X\pinv
				\Pro_{\X_i}
				\X_i
				&[\text{\cref{lemm:Ak_space}}]
				\\
				&=
				\T
				\, .
				&[\Pro_\A \A = \A]
			\end{aligned}$
		\item \label{item:2} $\bS \T \bS  = \bS$,
			as may be shown similarly to point \labelcref{item:1}.
		\item \label{item:3} $\T \bS = 
			\X\pinv \X
			$,
			as may be shown similarly to point \labelcref{item:1},
			using \cref{conj:rank_EnRML}.
			The symmetry of $\T \bS$ follows from that of
			$\X\pinv \X$.
		\item \label{item:4} The symmetry of $\bS \T$ is shown as for point \labelcref{item:3}.
			\vspace{-2em}
	\end{enumerate}
\end{proof}

This proof was heavily inspired by appendix A of \citet{sakov2012iterative}.
However, our developments apply for EnRML
(rather than the deterministic, square-root IEnKS).
This means that $\T_i$ is not symmetric, which complicates the proof
in that the focus must be on $ \X\pinv \X_i $ rather than $\X_i\pinv$ alone.
Our result also shows the equivalence of $\bS\pinv$ and $\T$ in general,
while the additional result of the vanishing projection matrix in
the case of $N-1 \leq M$ is treated separately, in \cref{sec:red_T}.

\subsection{Proof of \cref{eqn:Yred}}
\label{sec:red_T}

\begin{lemm}[]
	$\Omeg_i$ is invertible (provided $\W_i$ is).
	\label{lemm:Omeg}
\end{lemm}
\begin{proof}
	We show that
	$\Omeg_i \bu \neq 0$ for any $\bu \neq 0$,
	where $\Omeg_i = \W_i \PiAN + \PiOne$.
	For $\bu \in \col(\ones)$: $\Omeg_i \bu = \bu$.
	For $\bu \in \col(\ones)^\perp$:
	$\Omeg_i \bu = \W_i \bu \neq 0 $ (\cref{conj:rank_EnRML}).
\end{proof}

Recall that \cref{eqn:T_PWP} was obtained by inserting $\X_i$ in the expression \labelcref{eqn:T_def} for $\T_i$.
By contrast, the following inserts $\X$ from \cref{eqn:fawef} in the expression \labelcref{eqn:m_2} for $\T_i\pinv$,
yielding
$\T_i\pinv = \X_i\pinv \X = \X_i \X_i \Omeg_i^{-1} = \Pro_{\X_i\tr} \Omeg_i^{-1} = \PiAN \Pro_{\X_i\tr} \Omeg_i^{-1}$,
and hence 
\begin{align}
	\Y_i
	&=
	[\ObsMod(\E_i) \PiAN] \Pro_{\X_i\tr} \Omeg_i^{-1}
	\label{eqn:m_3}
	\, .
\end{align}
Next, it is shown that, under certain conditions,
the projection matrix $\Pro_{\X_i\tr}$ vanishes:
\begin{align}
	\Y_i
	&=
	[\ObsMod(\E_i) \PiAN] \Omeg_i^{-1}
	\label{eqn:m_4}
	\, .
\end{align}
Thereafter, \cref{eqn:Omeg_inv} of \cref{sec:pinv_v} can be used to
write $\Omeg_i^{-1}$ in terms of $(\W_i \PiAN)\pinv$,
reducing \cref{eqn:m_4} to \labelcref{eqn:Yred}.

\subsubsection*{The case of $\compactN \leq M$}
In the case of $\compactN \leq M$, the null space of $\X$ is the range of $\ones$ \citep[with probability 1,][Theorem 3.1.4]{muirhead1982aspects}.
By \cref{lemm:Omeg}, the same applies for $\X_i$, and so $\Pro_{\X_i\tr}$ in \cref{eqn:m_3} reduces to $\PiAN$.
$\square$

\subsubsection*{The case of linearity}
Let $\ObsMat$ be the matrix of the observation model $\ObsMod$, here assumed linear:
$\ObsMod(\E_i) = \ObsMat \E_i$.
By \cref{eqn:m_3},
$\Y_i = \ObsMat \E_i \Pro_{\X_i\tr} \Omeg_i^{-1}$.
But $\E_i \Pro_{\X_i\tr} = \X_i = \E_i \PiAN$.
$\square$

\subsection{Inverse transforms}
\label{sec:pinv_v}
Recall from \cref{eqn:Yk0_def} that $\Y_i \ones = 0$.
Therefore
\begin{align}
	\Cwk^{\pm 1} \ones
	&=
  \cN^{\mp 1} \ones
	\, ,
	\label{eqn:Cw_ones}
\end{align}
where
$\Cwk$ is defined in \cref{eqn:tP}, and
the identity for $\Cwk$ follows from that of $\Cwk^{-1}$.
Similarly, the following identities are valid
also when $\W_i$ and $\W_i^{-1}$ are swapped.
\begin{align}
	\W_i\tr \ones
	&=
	\ones
	\, ;
	\label{eqn:W_ones}
	\\
	\W_i \PiAN
	&=
	\PiAN \W_i \PiAN
	\, ;
	\label{eqn:W_proj}
	\\
	(\W_i \PiAN)\pinv
	&=
	\W_i^{-1} \PiAN
	\label{eqn:W_inv}
	\, .
\end{align}
\Cref{eqn:W_ones} is proven inductively (in $i$)
by inserting \labelcref{eqn:Cw_ones} in line \ref{ln:Wk} of \cref{algo:GN_EnRML}.
It enables showing \labelcref{eqn:W_proj}, using $\PiAN = \I_N - \PiOne$.
This enables showing \labelcref{eqn:W_inv}, similarly to \cref{lemm:T_pinv}.
Note that this implies that $\Y_i \ones = 0$
also for $\Y_i = \ObsMod(\E_i) \, (\W_i \PiAN)\pinv$,
and hence that the identities of this section also hold with this definition.
\Cref{eqn:W_proj,eqn:W_inv} can be used to show (by multiplying with $\Omeg_i$) that
\begin{align}
	\Omeg_i^{-1}
	&=
	(\W_i \PiAN)\pinv + \PiOne
  \, .
	\label{eqn:Omeg_inv}
\end{align}




\subsection*{Acknowledgements}                                                           
The authors thank Dean Oliver, Kristian Fossum, Marc Bocquet, and Pavel Sakov
for their reading and comments,
and Elvar Bjarkason for his questions concerning the computation
of the inverse transform matrix.
This work has been funded by DIGIRES,
a project sponsored by industry partners
and the PETROMAKS2 programme of the Research Council of Norway.

\bibliographystyle{plainnat}                                                            
\bibliography{localrefs}

\end{document}